\documentclass[conference]{IEEEtran}
\IEEEoverridecommandlockouts

\usepackage{amsmath,amssymb,amsfonts}
\usepackage{amsthm}
\usepackage{bm}
\usepackage{cite}
\usepackage{graphicx}
\usepackage{mathtools}
\usepackage{url}
\usepackage{algorithm}
\usepackage{algpseudocode}
\usepackage{tikz}
\usepackage[hidelinks]{hyperref}
\usetikzlibrary{shapes.geometric,arrows.meta,positioning}
\usepackage{pgfplots}
\pgfplotsset{compat=1.17}
\definecolor{figblue}{RGB}{37,99,235}
\definecolor{figgreen}{RGB}{5,150,105}
\definecolor{figpurple}{RGB}{124,58,237}
\definecolor{figred}{RGB}{225,29,72}

\makeatletter
\def\bstctlcite#1{\@bsphack
  \@for\@citeb:=#1\do{%
    \edef\@citeb{\expandafter\@firstofone\@citeb}%
    \if@filesw\immediate\write\@auxout{\string\citation{\@citeb}}\fi}%
  \@esphack}
\makeatother

\newtheorem{fact}{Fact}
\newtheorem{proposition}{Proposition}
\newtheorem{corollary}{Corollary}
\newtheorem{lemma}{Lemma}
\newtheorem{remark}{Remark}

\title{\LARGE \bf
Excitation-Supervised Closed-Loop Self-Calibration and Target Seeking for an Unknown-Pose Range-Bearing Relay}

\author{Yash Bagla\\
Michigan State University\\
{\tt\small yashbagla321@gmail.com}%
\thanks{\raggedright Code and data: \protect\url{https://github.com/yashbagla321/excitation-supervised-closed-loop}, archived at \protect\url{https://doi.org/10.5281/zenodo.21892671}.}}

\begin{document}
\bstctlcite{IEEEexample:BSTcontrol}
\maketitle
\thispagestyle{empty}
\pagestyle{empty}
\begin{abstract}
A vehicle seeking a hidden target through a range-bearing relay of unknown position and yaw must decide, online, whether its own motion has already made the relay calibration trustworthy, and what to do when it has not.
Two distinct vehicle-relative observations are known to remove the calibration gauge and make the target's relay-local packet globally actionable~\cite{bagla2027identifiability}, but that identifiability statement is static: it classifies a stored trajectory window only after the fact.
This paper supplies the closed-loop layer: we show that the trajectory-spread margin $S_v$ that governs identifiability is simultaneously a finite-noise seed-accuracy bound, a local-vector variance decomposition, and a circle-geometry excitation budget, and we use it to supervise an excitation-reset controller.
An excitation-supervised algorithm retriggers exploratory motion whenever the spread certificate is insufficient, projecting the target-seeking input away from the excitation's push direction while it does so, and otherwise proceeds to unrestricted target seeking; under explicit sampling assumptions the supervision rule provably acquires any required excitation within a finite time, while, in the noiseless local regime with positive excitation decay, estimator convergence yields target-seeking convergence after certification, and the threshold itself is selected from a desired calibration-accuracy level rather than chosen heuristically.
Closed-loop simulation, paired Monte Carlo comparisons, a spread-threshold ablation, and a ROS~2/Gazebo software-in-the-loop experiment, in which the supervised loop steers a physics-integrated vehicle through real message transport and sensing delay, validate the approach.
A decay-rate sweep against a fixed decaying-excitation schedule shows that supervision matters when the schedule's decay outruns the unknown time-to-adequate-excitation: over 100 paired trials the fixed baseline's yaw RMSE rises from $0.010$ to $0.065$~rad and success falls to $56\%$, while target-tracking error remains insensitive; supervision keeps yaw RMSE between $0.0095$ and $0.0191$~rad with $100\%$ success.
\end{abstract}

\section{Introduction}

Vehicle motion can self-calibrate an unknown-pose range-bearing relay: two distinct vehicle-relative observations are enough to recover the relay's yaw, position, and a hidden target it reports only in its own local frame~\cite{bagla2027identifiability}.
That identifiability result is static in the sense that it classifies a finite trajectory window as sufficient or insufficient after the fact.
A vehicle running target-seeking control online instead needs to decide, at each measurement step, whether its stored window is already excited enough to trust the current estimate, or whether it should keep exploring before target seeking dominates.
This paper treats that supervision problem as a control-design question built on top of the identifiability result, rather than re-deriving identifiability itself.

The trajectory-spread margin $S_v$ that governs identifiability conditioning in~\cite{bagla2027identifiability} turns out to carry three additional roles once the estimator is placed inside a control loop: it bounds the constructive seed's sensitivity to measurement noise, it decomposes the target-estimate covariance into an averaging term and a calibration-propagation term, and, for circular excitation, it has a closed form in the excitation radius and packet count.
Together these give a controller a certifiable, checkable stopping rule for exploratory motion, rather than a fixed exploration schedule chosen in advance.

The present paper retains the globally referenced vehicle trajectory of~\cite{bagla2027identifiability} in order to isolate how adequate excitation is acquired and certified online.
A complementary anchor-free formulation~\cite{bagla2026anchorfree} removes this assumption: it preserves the resulting global $SE(2)$ gauge, estimates only the control-relevant quotient from integrated odometry, and accounts for the cross-view pose correlations that then dominate calibration uncertainty.
Thus, the two closed-loop formulations address distinct regimes---excitation supervision with globally referenced ego motion here, and certified quotient control under correlated odometry in the anchor-free setting.

Persistency of excitation is the classical analogue in adaptive control: parameter estimates converge only when the regressor signal is spectrally rich enough~\cite{boyd1986necessary,narendra1987persistent}, and modern estimator designs such as dynamic regressor extension and mixing work to weaken exactly this requirement~\cite{aranovskiy2017drem}.
The excitation margin $S_v$ plays a similar role for relay self-calibration, but as a scalar richness condition a supervisory controller checks online from stored packets, rather than a spectral condition imposed on a reference signal offline.
Concurrent learning replaces the offline persistent-excitation requirement with a verifiable online condition on recorded data, using a similar motivation to the one behind Algorithm~\ref{alg:supervised}: monitor a checkable richness condition rather than assume it holds~\cite{chowdhary2013concurrent}.
The rotating exploratory envelope that Algorithm~\ref{alg:supervised} retriggers plays the same role as the periodic dither that extremum-seeking control injects to keep its gradient information alive~\cite{krstic2000stability}, and as the excitation used in adaptive source-localization station keeping~\cite{guler2017adaptive}.
Observability-aware trajectory optimization designs motion to maximize an observability metric for self-calibration states, typically for inertial or visual-inertial sensor suites~\cite{hausman2017observability}, including closed-loop navigation and online recalibration~\cite{preiss2018simultaneous}, while closed-form circular or spiral excitation has been proposed specifically for bearing-only target localization and sensor self-calibration~\cite{peng2024trajectory}; the excitation-budget corollary below is a closed-form analogue for the relay-frame gauge, in the tradition of optimal sensor--target localization geometries~\cite{bishop2010optimality,martinez2006optimal,zhou2011multirobot}.
More broadly, active SLAM and information-based exploration treat sensing and self-localization jointly through information-theoretic or optimal-experimental-design criteria~\cite{bourgault2002information,placed2023active}, and dual-control formulations make the underlying probe-versus-exploit tradeoff explicit in greater generality~\cite{mesbah2018stochastic}. Algorithm~\ref{alg:supervised} takes a complementary route: it replaces online trajectory optimization and matrix-valued information objectives with a relay-specific scalar certificate whose threshold is derived from calibration accuracy and whose acquisition under sampled actuation is proved below.

The contributions are:
\begin{enumerate}
\item a quantitative excitation theory showing that the spread margin $S_v$ simultaneously bounds constructive-seed error, decomposes target-estimate variance into averaging and calibration-propagation terms, and gives an explicit excitation budget for circular motion, together with an accuracy-driven rule selecting the spread threshold from a desired calibration-uncertainty level;
\item an excitation-supervised estimation and target-seeking algorithm that formalizes the informal reset rule of~\cite{bagla2027identifiability} into an explicit online supervision loop triggered by the spread certificate alone;
\item a finite excitation-acquisition proposition showing that, under explicit sampling assumptions, the supervision rule itself accumulates any required spread within an explicit finite time, with every hypothesis satisfied by the controller as implemented through an underexcited-phase projection of the seeking input, and a local recovery proposition showing that, in the noiseless model and for positive excitation decay, entry into the Gauss--Newton convergence basin yields target-estimate and target-seeking convergence;
\item a ROS~2/Gazebo software-in-the-loop experiment in which the supervised loop, running as a ROS~2 node, steers a vehicle integrated by the Gazebo physics engine through real message transport and configurable sensing delay;
\item a 100-trial paired Monte Carlo decay-rate sweep and a spread-threshold ablation that validate the accuracy-driven threshold rule as a population root-mean-square-error target, quantify its excitation cost, and expose its domain of validity, together with a target-seeking scenario showing that ordinary closed-loop motion can supply the needed excitation naturally, so that supervision matters most when such incidental excitation is absent.
\end{enumerate}

\section{Measurement Model and Identifiability Recap}
\label{sec:model}

We recap the single-beacon model of~\cite{bagla2027identifiability}: a hidden target $p\in\mathbb{R}^2$ is reported by a relay beacon (the two terms are used interchangeably) of unknown position $x\in\mathbb{R}^2$ and yaw $\psi$ (rotation $R=R(\psi)$), which returns, at each vehicle pose $q_k$, noisy local-frame range-bearing packets to the vehicle and to the target,
\[
\tilde\ell_{k}^{v}=\tilde r_{k}^{v}
\begin{bmatrix}\cos\tilde\beta_{k}^{v}\\ \sin\tilde\beta_{k}^{v}\end{bmatrix},\qquad
\tilde\ell_k^{t}=\tilde r_k^{t}
\begin{bmatrix}\cos\tilde\beta_k^{t}\\ \sin\tilde\beta_k^{t}\end{bmatrix},
\]
with noiseless vectors $\ell_{k}^{v}$ and $\ell^{t,*}$ satisfying
\begin{equation}
q_k=x+R\ell_{k}^{v},\qquad
p=x+R\ell^{t,*}. \label{eq:model}
\end{equation}
The state is $z=(p,x,\psi)$.
As in~\cite{bagla2027identifiability}, the estimator minimizes native range-bearing residuals $r_k(z)=y_k-g_k(z)$ (with $y_k=(\tilde r_k^v,\tilde\beta_k^v,\tilde r_k^t,\tilde\beta_k^t)$ the stacked measurements at pose $q_k$ and $g_k(z)$ their model predictions under \eqref{eq:model}), weighted by the range/bearing precision matrix $W=\mathrm{diag}(\sigma_r^{-2},\sigma_\theta^{-2},\sigma_r^{-2},\sigma_\theta^{-2})$, in the batch objective
\begin{equation}
J(z)=\frac{1}{2}\sum_{k=1}^{K} r_{k}(z)^TWr_{k}(z). \label{eq:cost}
\end{equation}

\begin{fact}[Motion-induced two-view recovery~\cite{bagla2027identifiability}]
\label{fact:constructive}
Consider noiseless measurements at two vehicle poses $q_a,q_b$.
If $\ell_a^v\neq\ell_b^v$, then $p$, $x$, and $\psi$ are uniquely determined modulo $2\pi$ yaw wrapping by
\[
\psi=\angle(q_b-q_a)-\angle(\ell_b^v-\ell_a^v),
\]
\[
x=q_a-R\ell_a^v,\qquad
p=x+R\ell^{t,*},
\]
and the Jacobian of the stacked residuals in \eqref{eq:cost} has rank five at the true state if and only if $\{\ell_k^v\}$ contains at least two distinct vectors.
A single vehicle pose leaves a one-dimensional yaw/translation/target gauge.
\end{fact}

Fact~\ref{fact:constructive} is a static classification: a given window either does or does not contain two distinct local vehicle vectors.
The rest of this paper is about what a controller should do with that classification online.

\section{Quantitative Excitation Theory}
\label{sec:excitation}

\begin{lemma}[Finite-window excitation margin~\cite{bagla2027identifiability}]
\label{lem:spread}
Define the centered local vehicle-vector spread
\[
S_v=\sum_{k=1}^{K}\|\ell_k^v-\bar\ell^v\|^2,\qquad
\bar\ell^v=\frac{1}{K}\sum_{k=1}^{K}\ell_k^v .
\]
For any fixed yaw perturbation $\delta\psi$,
$\min_{\delta x}\sum_k\|\delta x+R S\ell_k^v\delta\psi\|^2=\delta\psi^2S_v$, where $S=\left[\begin{smallmatrix}0&-1\\1&0\end{smallmatrix}\right]$.
In the noiseless model, $S_v=\sum_k\|q_k-\bar q\|^2$.
\end{lemma}

The identifiability analysis of~\cite{bagla2027identifiability} uses $S_v$ to explain estimator conditioning and states an informal retrigger rule on it, without quantifying what a given spread level buys.
The following results supply that quantification: $S_v$ certifies, before the estimate is trusted, how much a controller should expect the constructive seed and the closed-loop estimate to be worth.

\begin{proposition}[Constructive-seed accuracy]
\label{prop:seed}
Let the two selected vehicle packets satisfy $\|\tilde\ell_j^v-\ell_j^v\|\le\epsilon$, $j\in\{a,b\}$, and let each of the $K$ target packets satisfy $\|\tilde\ell_k^t-\ell^{t,*}\|\le\epsilon$.
Define $\bar{\tilde\ell}^{\,t}=K^{-1}\sum_k\tilde\ell_k^t$, let $d_v=\|\ell_b^v-\ell_a^v\|$ with $4\epsilon\le d_v$, and form the target seed as $\hat p=\hat x+R(\hat\psi)\bar{\tilde\ell}^{\,t}$.
Then the resulting constructive estimates satisfy
\[
|\hat\psi-\psi|\le\frac{2\pi\epsilon}{d_v},\quad
\|\hat x-x\|\le\epsilon\!\left(1+\frac{2\pi r_a}{d_v}\right),
\]
\[
\|\hat p-p\|\le\epsilon\!\left(2+\frac{2\pi(r_a+r_t)}{d_v}\right),
\]
where $r_a=\|\ell_a^v\|$ and $r_t=\|\ell^{t,*}\|$.
Since $d_v^2=2S_v$ for $K=2$, seed error scales as measurement noise divided by the square root of trajectory spread.
\end{proposition}

The worst-case bound does not improve with $K$, whereas under zero-mean noise the benefit of averaging appears through the $\sigma^2/K$ term in Corollary~\ref{cor:variance}.

\begin{proof}
Write $v=\ell_b^v-\ell_a^v$ and $\tilde v=v+w$ with $\|w\|\le2\epsilon$, so the yaw error $|\hat\psi-\psi|$ equals the angle $\alpha$ between $v$ and $\tilde v$; since $\|w\|\le2\epsilon\le d_v/2<\|v\|$, the inner product $v^T\tilde v$ is positive and $\alpha\in[0,\pi/2)$.
Bounding $\sin\alpha$ and then $\alpha$ itself gives
\begin{align}
\sin\alpha&=\frac{|v\times w|}{\|v\|\|\tilde v\|}\le\frac{2\epsilon}{d_v-2\epsilon}\le\frac{4\epsilon}{d_v}, \label{eq:sinalpha}\\
|\hat\psi-\psi|=\alpha&\le\frac{\pi}{2}\sin\alpha\le\frac{2\pi\epsilon}{d_v}, \label{eq:alphabound}
\end{align}
where \eqref{eq:sinalpha} uses $d_v\ge4\epsilon$ so $d_v-2\epsilon\ge d_v/2$, and \eqref{eq:alphabound} uses $\alpha\le(\pi/2)\sin\alpha$ on $[0,\pi/2]$.
Propagating \eqref{eq:alphabound} through $x=q_a-R(\psi)\ell_a^v$ gives
\begin{align}
\|\hat x-x\|&\le\underbrace{\|\tilde\ell_a^v-\ell_a^v\|}_{\le\,\epsilon}+\underbrace{\|(R(\hat\psi)-R(\psi))\ell_a^v\|}_{\le\,|\hat\psi-\psi|\,r_a} \nonumber\\
&\le\epsilon\Big(1+\frac{2\pi r_a}{d_v}\Big). \label{eq:xbound}
\end{align}
Because $\|\bar{\tilde\ell}^{\,t}-\ell^{t,*}\|\le K^{-1}\sum_k\|\tilde\ell_k^t-\ell^{t,*}\|\le\epsilon$, the same two-term split applied to $p=x+R(\psi)\ell^{t,*}$, using \eqref{eq:xbound} for the $x$ contribution and $r_t=\|\ell^{t,*}\|$ for the rotation contribution, gives the stated bound on $\|\hat p-p\|$.
\end{proof}

\begin{corollary}[Local-vector variance decomposition]
\label{cor:variance}
Consider the linearized local-vector model at the true one-beacon state with i.i.d. isotropic noise of variance $\sigma^2$ per axis on each vehicle and target packet, $K$ synchronized packets, and a flat prior on $z=(p,x,\psi)$.
Marginalizing $p$ removes the target rows' Fisher information about $(x,\psi)$; the information on $(x,\psi)$ comes from the vehicle packets alone, and the yaw Schur complement equals $S_v/\sigma^2$.
Consequently $\mathrm{var}(\hat\psi)=\sigma^2/S_v$ in the linearized model, and
\[
\mathrm{cov}(\hat p)=\frac{\sigma^2}{K}I+B F_\zeta^{-1}B^T,\qquad
B=\begin{bmatrix}I&RS\ell^{t,*}\end{bmatrix},
\]
where $F_\zeta$ is the vehicle-only information on $\zeta=(\delta x^T,\delta\psi)^T$~\cite{barshalom2001estimation}.
The first term is target-packet averaging; the second is calibration error propagated through the relay frame and controlled by $S_v$.
\end{corollary}

\begin{proof}
The target rows have variation $\delta p-B\zeta$ plus noise, so their joint Fisher information on $(p,\zeta)$ over $K$ packets is
\[
\mathcal I_{\rm target}=\frac{K}{\sigma^2}
\begin{bmatrix}I&-B\\-B^T&B^TB\end{bmatrix}.
\]
Eliminating $p$ by the Schur complement gives
\[
B^TB-B^T(I)^{-1}B=0,
\]
so the target rows alone carry no information on $\zeta$ once $p$ is free to absorb their variation.
The vehicle rows give joint information on $(x,\psi)$
\[
\mathcal I_{\rm vehicle}=\frac{1}{\sigma^2}
\begin{bmatrix}
KI & RS\sum_k\ell_k^v\\
\big(RS\sum_k\ell_k^v\big)^T & \sum_k\|\ell_k^v\|^2
\end{bmatrix}.
\]
Eliminating $x$ by the Schur complement and using $RS\sum_k\ell_k^v=KRS\bar\ell^v$ leaves the yaw information
\[
F_\psi=\frac{1}{\sigma^2}\Big(\textstyle\sum_k\|\ell_k^v\|^2-K\|\bar\ell^v\|^2\Big)=\frac{S_v}{\sigma^2},
\]
using $\|RS\bar\ell^v\|=\|\bar\ell^v\|$ and the centering identity of Lemma~\ref{lem:spread}.
Since the target rows contribute zero information on $\zeta$, $F_\zeta$ is exactly $\mathcal I_{\rm vehicle}$, and $\hat p$ decomposes as a two-stage estimator: the target-packet average contributes $(\sigma^2/K)I$ when $\zeta$ is known, and propagating the calibration estimate $\hat\zeta$ through the linear map $B$ adds $BF_\zeta^{-1}B^T$, because the two error sources are informationally orthogonal by the vanishing cross term above.
\end{proof}

Proposition~\ref{prop:seed} and Corollary~\ref{cor:variance} give the controller a quantitative reason to keep exploring when $S_v$ is small: seed error and calibration-propagated target variance both scale with $1/\sqrt{S_v}$ or $1/S_v$, not just with the binary rank condition of Fact~\ref{fact:constructive}.

\begin{corollary}[Excitation budget]
\label{cor:budget}
Let the stored window contain $K$ vehicle positions on a circle of radius $\rho$, $q_k=c+\rho u_k$ with $\|u_k\|=1$, and let $\bar u=K^{-1}\sum_k u_k$.
Then, in the noiseless model,
\[
S_v=K\rho^2(1-\|\bar u\|^2).
\]
For $K\ge2$ equally spaced positions over a full revolution, $\bar u=0$ and $S_v=K\rho^2$ exactly.
Thus the threshold $\bar S$ is met after at most $\lceil\bar S/\rho^2\rceil$ equally spaced poses, giving an explicit tradeoff between excitation radius, packet count, and calibration time.
\end{corollary}

\begin{proof}
By Lemma~\ref{lem:spread}, $S_v=\sum_k\|q_k-\bar q\|^2$ with $\bar q=c+\rho\bar u$.
Thus $S_v=\rho^2\sum_k\|u_k-\bar u\|^2=\rho^2(K-K\|\bar u\|^2)$.
Equally spaced unit vectors over a full revolution sum to zero.
\end{proof}

\begin{remark}[Accuracy-driven threshold selection]
\label{rem:threshold}
In the noiseless model, rank recovery requires only $S_v>0$, so no fixed positive threshold is universal.
A principled $\bar S$ follows instead from an accuracy requirement: Corollary~\ref{cor:variance} gives $\mathrm{var}(\hat\psi)=\sigma^2/S_v$, so a required yaw standard deviation $\epsilon_\psi$ induces
\begin{equation}
\bar S\ge\frac{\sigma^2}{\epsilon_\psi^2},
\qquad\text{hence}\qquad
K\ge\frac{\sigma^2}{\rho^2\,\epsilon_\psi^2} \label{eq:designrule}
\end{equation}
under equally spaced circular excitation by Corollary~\ref{cor:budget}.
A calibration-accuracy requirement therefore translates directly into an excitation radius and packet count, and the supervisor's threshold becomes an accuracy-driven design quantity rather than a heuristic constant.
\end{remark}

\begin{table}[t]
\centering
\caption{Excitation budget of Corollary~\ref{cor:budget}.}
\label{tab:budget}
\scriptsize
\begin{tabular}{l c c}
\hline
Radius $\rho$ (m) & $K$ for $\bar S=9.04$ & $K$ for $\bar S=100$\\
\hline
0.5 & 37 & 400\\
1.0 & 10 & 100\\
1.5 & 5 & 45\\
2.0 & 3 & 25\\
\hline
\end{tabular}\\[3pt]
\parbox{\columnwidth}{\raggedright\scriptsize Note: $K=\lceil\bar S/\rho^2\rceil$ equally spaced poses reach spread threshold $\bar S$ at excitation radius $\rho$.}
\end{table}

Table~\ref{tab:budget} instantiates the design rule \eqref{eq:designrule} at the simulation noise level $\sigma=0.02$~m: $\bar S=9.04$ enforces $\epsilon_\psi\approx6.7$~mrad and $\bar S=100$ enforces $\epsilon_\psi=2$~mrad.
Neither value is a universal rank threshold (noiseless identifiability needs only $S_v>0$); both are accuracy targets, and $\bar S=9.04$ also coincides with the two-pose spread used in the conditioning study of~\cite{bagla2027identifiability}.
For native range-bearing packets, $\sigma_{\rm eff}$ below is a conservative first-order proxy for the local-vector noise, giving the rule an operating-set interpretation rather than an isotropic-noise requirement.
If all vehicle--relay ranges obey $0<r_k^v\le r_{\max}$, the native-information bound of~\cite{bagla2027identifiability} gives
\begin{equation}
\sigma_{\rm eff}=\max\{\sigma_r,r_{\max}\sigma_\theta\},\qquad
\bar S\ge\frac{\sigma_{\rm eff}^2}{\epsilon_\psi^2}. \label{eq:nativethreshold}
\end{equation}
Indeed the radial and tangential local-vector information eigenvalues are $\sigma_r^{-2}$ and $(r_k^v\sigma_\theta)^{-2}$, so the yaw Schur complement is at least $S_v/\sigma_{\rm eff}^2$.
Equation~\eqref{eq:nativethreshold} therefore converts bearing noise into its worst-case transverse position uncertainty over the declared range envelope; the isotropic formula is recovered when $\sigma\ge r_{\max}\sigma_\theta$.
Doubling the excitation radius quarters the required pose count at fixed threshold, since $K=\lceil\bar S/\rho^2\rceil$; a controller with authority to choose $\rho$ can therefore trade a larger excitation loop for fewer packets and less calibration time, or a tighter loop for more packets, without any new simulation, just the closed form above.
The closed-loop ablation of Section~\ref{sec:ablation} probes the complementary boundary: a threshold beyond the spread reachable at the loop's actual excitation radius and packet budget degenerates to continuous excitation, so large $\bar S$ values are meaningful only together with the radius the budget prescribes for them.

\section{Excitation-Supervised Estimation and Target Seeking}
\label{sec:control}

Estimation uses the Gauss--Newton implementation of~\cite{bagla2027identifiability}, with analytic Jacobians for the residuals in \eqref{eq:cost} and Levenberg damping.
The target estimate remains at its prior until the stored window contains two distinct views; the implementation then applies the constructive seed once and warm-starts every later packet from the previous refined estimate, with damping initialized at $10^{-2}$ and at most 25 iterations per packet.
The vehicle obeys $\dot q=u$ with
\begin{align}
u&=u^{\rm seek}(t)+u^{\rm exp}(t), \nonumber\\
u^{\rm exp}(t)&=Ae^{-\lambda(t-t_0)}
\begin{bmatrix}\cos\omega t\\ \sin\omega t\end{bmatrix}, \label{eq:controller}
\end{align}
where $u^{\rm seek}$ is the nominal target-seeking law $-k(q-\hat p)$, projected while the stored window is underexcited so that it can never cancel the excitation's push:
\begin{equation}
u^{\rm seek}=\begin{cases}
\Pi_{\mathcal H(t)}\big[-k(q-\hat p)\big], & S_v(\mathcal{W})<\bar S,\\[1pt]
-k(q-\hat p), & S_v(\mathcal{W})\ge\bar S,
\end{cases}
\label{eq:projection}
\end{equation}
with $\mathcal H(t)=\{v:v^Tn(t)\ge-Ae^{-\lambda\bar T}/\pi\}$ and $n(t)=(-1)^{\lfloor\omega t/\pi\rfloor}[0,\ 1]^T$ the push direction of the current excitation half-period ($\bar T$ is the inter-packet bound of Proposition~\ref{prop:acquisition}).
Because the excitation phase is absolute, $n(t)$ is a causal, implementable quantity: half-periods tile at $\omega t=j\pi$, and $n$ survives epoch resets since it does not depend on $t-t_0$.
The projection $\Pi_{\mathcal H}$ clips only the seeking component opposing $n$, leaves lateral and aligned motion untouched, and the nominal law is restored exactly once the certificate clears; it is what lets the finite-acquisition guarantee below cover the controller as implemented rather than an idealized variant.
The excitation epoch $t_0$ resets whenever the stored packet window $\mathcal{W}$ is not yet excited enough; Algorithm~\ref{alg:supervised} makes this check explicit using the spread certificate $S_v(\mathcal{W})$ alone, computed from the known vehicle poses of the stored window, so both the reset trigger and the projection window are exactly the quantity covered by the acquisition guarantee of Proposition~\ref{prop:acquisition}.
The implementation uses the physical packet time $t_k=(k-1)\Delta t$, not the packet index, and applies the zero-order-hold update
\begin{align}
q_{k+1}={}&q_k+\Delta t\{u^{\rm seek}_k \nonumber\\
&+Ae^{-\lambda(t_k-t_0)}[\cos\omega t_k,\ \sin\omega t_k]^T\}, \label{eq:sampledcontroller}
\end{align}
with $u^{\rm seek}_k$ the projected law \eqref{eq:projection} evaluated at $(q_k,\hat p_k,t_k)$.
Thus $\lambda$ is reported in s$^{-1}$, $\omega$ in rad/s, and the batch and Gazebo implementations have the same phase, decay, and projection semantics.

Figure~\ref{fig:flowchart} previews the supervision loop that Algorithm~\ref{alg:supervised} makes precise.

\begin{figure}[t]
\centering
\resizebox{0.92\linewidth}{!}{%
\begin{tikzpicture}[
  node distance=7mm and 8mm,
  box/.style={rectangle, rounded corners, draw, fill=blue!6, align=center, minimum height=7mm, text width=30mm, font=\footnotesize},
  dec/.style={diamond, aspect=2.2, draw, fill=orange!8, align=center, inner sep=1pt, text width=26mm, font=\footnotesize},
  arr/.style={-{Latex[length=2mm]}, thick}
]
\node[box] (measure) {Measure packet; update $S_v(\mathcal{W})$};
\node[dec, below=of measure] (dec1) {Two distinct\\ views in $\mathcal{W}$?};
\node[box, below=of dec1] (seed) {Averaged constructive seed + Gauss--Newton refine};
\node[dec, below=of seed] (dec2) {$S_v(\mathcal{W})\ge\bar S$?};
\node[box, below=of dec2] (seek) {Target seeking: $u=-k(q-\hat p)+u^{\rm exp}(t)$};
\node[box, right=16mm of dec1, text width=26mm] (explore) {Retrigger: reset $t_0$; apply \eqref{eq:projection} $+\,u^{\rm exp}(t)$};

\draw[arr] (measure) -- (dec1);
\draw[arr] (dec1) -- node[left,font=\scriptsize]{yes} (seed);
\draw[arr] (dec1) -- node[above,font=\scriptsize]{no} (explore);
\draw[arr] (seed) -- (dec2);
\draw[arr] (dec2) -- node[left,font=\scriptsize]{yes} (seek);
\draw[arr] (dec2.east) -- ++(8mm,0) node[above,font=\scriptsize]{no} -| (explore.south);
\draw[arr] (explore.north) -- ++(0,10mm) |- (measure.east);
\draw[arr] (seek.south) -- ++(0,-6mm) -| ++(-45mm,0) |- (measure.west);
\end{tikzpicture}}
\caption{Supervision loop of Algorithm~\ref{alg:supervised}: exploratory motion is retriggered whenever the stored window's spread certificate is below threshold, with the seeking input projected so it cannot oppose the excitation, and unrestricted target seeking proceeds once the certificate clears.}
\label{fig:flowchart}
\end{figure}
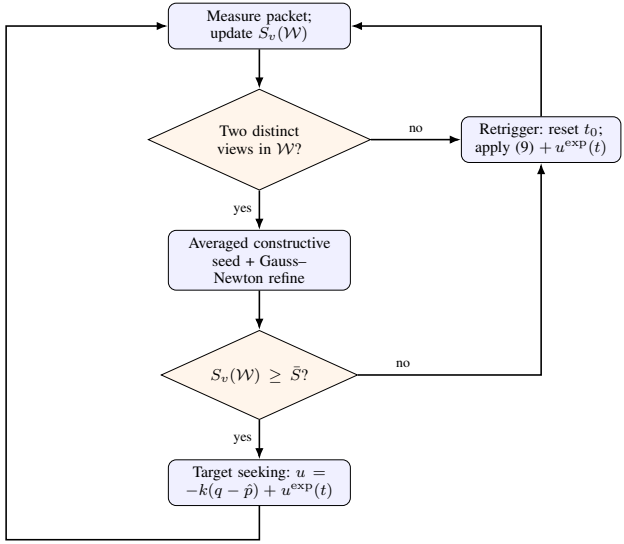

\begin{algorithm}[t]
\caption{Excitation-supervised self-calibration and target seeking}
\label{alg:supervised}
\begin{algorithmic}[1]
\Require gains $k,A,\lambda,\omega$; spread threshold $\bar S>0$; initial target estimate $\hat p_0$
\State $\mathcal{W}\leftarrow\emptyset$, $t_0\leftarrow0$, $\hat p\leftarrow\hat p_0$, initialized$\leftarrow$false
\For{each packet $k$ received at time $t$}
  \State $\mathcal{W}\leftarrow \mathcal{W}\cup\{(q_k,\text{packet})\}$
  \State update $S_v(\mathcal{W})$
  \If{not initialized and $\mathcal{W}$ contains two distinct observations}
    \State compute $(\hat\psi,\hat x)$ from the two-view formulas in Fact~\ref{fact:constructive}
    \State set $\hat p\leftarrow\hat x+R(\hat\psi)\bar{\tilde\ell}^{\,t}$ using the target-packet mean over $\mathcal{W}$
    \State initialized$\leftarrow$true
  \EndIf
  \If{initialized}
    \State refine $(\hat\psi,\hat x,\hat p)$ by warm-started damped Gauss--Newton~\cite{nocedal2006numerical}
  \EndIf
  \If{$S_v(\mathcal{W})<\bar S$}
    \State $t_0\leftarrow t$ \Comment{retrigger excitation envelope \eqref{eq:controller}}
  \EndIf
  \State apply $u=u^{\rm seek}+u^{\rm exp}(t)$ with the projected seeking law \eqref{eq:projection}
\EndFor
\end{algorithmic}
\end{algorithm}

The recovery result below is conditional on the stored window eventually meeting the excitation threshold.
The next proposition discharges that condition: the supervision rule itself acquires any required spread in explicit finite time, so retriggering is not merely a heuristic hope but a guaranteed acquisition mechanism.

\begin{proposition}[Finite acquisition of excitation]
\label{prop:acquisition}
Let $A>0$, $\omega>0$, $\lambda\ge0$, and consider Algorithm~\ref{alg:supervised} with the controller \eqref{eq:controller}--\eqref{eq:projection} in the noiseless model, applied with zero-order hold as in \eqref{eq:sampledcontroller}: the command is computed at each stored packet and held until the next.
Suppose that during underexcited operation ($S_v(\mathcal{W})<\bar S$):
(i) packets are stored with inter-sample times in $[\underline T,\bar T]$, $0<\underline T\le\bar T$;
(ii) the seeking input is the projected law \eqref{eq:projection}, so $u^{\rm seek}(t)^Tn(t)\ge-Ae^{-\lambda\bar T}/\pi$ by construction;
(iii) sampling is fast relative to the excitation, $8\omega\bar T\le e^{-\lambda\bar T}$;
(iv) the window $\mathcal{W}$ is cumulative: every stored packet is retained, none is evicted.
Then for every threshold $\bar S>0$, with $\delta:=Ae^{-\lambda\bar T}/(2\omega)$, the stored window satisfies $S_v(\mathcal{W})\ge\bar S$ within time
\[
T^\star\le\frac{\pi}{\omega}\Big(\Big\lceil\frac{2\bar S}{\delta^2}\Big\rceil+2\Big)
\]
of the start of underexcited operation, and Algorithm~\ref{alg:supervised} performs at most $\lfloor T^\star/\underline T\rfloor+1$ spread-triggered resets within that window.
\end{proposition}

\begin{proof}
While $S_v(\mathcal{W})<\bar S$, Algorithm~\ref{alg:supervised} resets $t_0\leftarrow t$ at every stored packet before the command is computed, so each hold interval $[t_k,t_{k+1})$ of \eqref{eq:sampledcontroller} carries the excitation at full amplitude with frozen phase, $u^{\rm exp}\equiv A[\cos\omega t_k,\ \sin\omega t_k]^T$; write $A_-:=Ae^{-\lambda\bar T}$ for the decayed level that sets the seeking allowance and $\delta$.
The interval before the first packet lasts at most $\bar T\le1/(8\omega)$ by (iii) and is absorbed into the alignment slack below.
Fix a half-period $[s,s+\pi/\omega]$ of underexcited operation, on which $n(t)\equiv n_s=[-\sin\omega s,\ \cos\omega s]^T$.
For a packet $t_k$ in the half-period, the held excitation pushes along $n_s$, $u^{\rm exp}(t_k)^Tn_s=A\sin(\omega(t_k-s))\ge0$, and the held seeking input obeys $u^{\rm seek}(t_k)^Tn_s\ge-A_-/\pi$ by (ii); both bounds persist for as long as the command is held, and every command active on $[t_1,t_2]$ below was computed at a packet time inside the same half-period.
By (i) there are stored packets $P_1$ taken at some $t_1\in[s,s+\bar T]$ and $P_2$ at some $t_2\in[s+\pi/\omega-\bar T,\ s+\pi/\omega]$; both lie inside the half-period, and they are distinct packets since $t_2-t_1\ge\pi/\omega-2\bar T>0$ by (iii).
The displacement between them accumulates the held commands.
The excitation part is the left-endpoint sum $\sum_kA\sin(\omega(t_k-s))\,(t_{k+1}-t_k)$ over $[t_1,t_2]$, which differs from $\int_{t_1}^{t_2}A\sin(\omega(t-s))\,dt$ by at most $(\bar T/2)(t_2-t_1)A\omega\le\pi A\bar T/2$ (the integrand's slope is at most $A\omega$ and each hold lasts at most $\bar T$); the integral itself is at least $2A/\omega-A\omega\bar T^2$ (each clipped edge costs at most $A\omega\bar T^2/2$, since $\sin\omega\tau\le\omega\tau$ there).
The held seeking part contributes at least $-(A_-/\pi)(t_2-t_1)\ge-A_-/\omega$, with no discretization error since its bound holds per hold interval.
Hence
\begin{equation}
(P_2-P_1)^Tn_s
\ge\frac{2A}{\omega}-A\omega\bar T^2-\frac{\pi A\bar T}{2}-\frac{A_-}{\omega}
\ge\frac{A_-}{2\omega}=\delta,
\end{equation}
where the second inequality is $(\omega\bar T)^2+\tfrac{\pi}{2}\omega\bar T\le2-\tfrac32e^{-\lambda\bar T}$, which (iii) implies with margin: $8\omega\bar T\le e^{-\lambda\bar T}\le1$ (using $\lambda\ge0$) gives $\omega\bar T\le1/8$, so the left side is at most $1/64+\pi/16<1/2$ while the right side is at least $1/2$.
Hence $\|P_1-P_2\|\ge\delta$, and consecutive half-periods contribute disjoint pairs: each pair lives inside its own half-period, and at most one stored packet can sit on a shared boundary instant (inter-sample times are positive), so representatives can be chosen distinct.
In the noiseless model $S_v=\sum_k\|q_k-\bar q\|^2=\min_c\sum_k\|q_k-c\|^2$ by Lemma~\ref{lem:spread}; dropping all non-pair terms and minimizing each pair separately at its midpoint,
\[
S_v\ge\sum_{i=1}^{N}\frac{\|P_1^{(i)}-P_2^{(i)}\|^2}{2}\ge\frac{N\delta^2}{2}.
\]
Thus $S_v\ge\bar S$ once $N\ge2\bar S/\delta^2$ half-periods of length $\pi/\omega$ have elapsed, plus at most two half-periods of alignment slack (one to reach the first half-period boundary, one absorbing the pre-first-packet interval).
By (iv) the window only grows, and $S_v$ is nondecreasing as it does (the minimum over $c$ can only increase when nonnegative terms are added), so the threshold, once met, stays met, and the spread trigger causes no further resets.
Each reset occurs at a stored packet, and a window of length $T^\star$ contains at most $\lfloor T^\star/\underline T\rfloor+1$ stored packets (including the one at its start), so at most that many resets occur.
\end{proof}

The bound credits only two stored packets per half-period, so it is conservative relative to simulated acquisition rates; its value is the logic chain it completes: supervision acquires excitation in finite time (Proposition~\ref{prop:acquisition}), excitation removes the gauge (Fact~\ref{fact:constructive}), and the recovery guarantee below then applies with its finitely-many-resets hypothesis discharged. Proposition~\ref{prop:acquisition} permits $\lambda=0$, whereas asymptotic target recovery additionally requires $\lambda>0$ so the excitation vanishes after the final reset.
Because Algorithm~\ref{alg:supervised} triggers on $S_v$ alone, the acquisition guarantee covers the quantity the controller actually uses; Proposition~\ref{prop:acquisition} does not require evaluating a coordinate-dependent matrix condition number online.

The implemented parameters make every hypothesis checkable, and the released per-packet logs make the check empirical rather than nominal.
With $k=1.2$~s$^{-1}$, $A=0.25$~m/s, $\omega=0.45$~rad/s, and $\bar T=\underline T=\Delta t=0.08$~s, hypothesis (i) holds exactly (fixed sampling), (ii) holds by construction because the projection \eqref{eq:projection} is part of the implemented controller, (iv) holds by construction (the stored window only grows), and (iii) holds across the entire decay sweep: $8\omega\bar T=0.288$, while $e^{-\lambda\bar T}$ ranges from $0.998$ at $\lambda=0.02$ down to $0.852$ at $\lambda=2$, so the sampling condition would first fail beyond $\lambda\approx15.6$~s$^{-1}$, far outside the tested range.
The projection is the step that closes the gap between theorem and controller, and it is not vacuous.
An unprojected seeking term at the implemented gain would exceed the $Ae^{-\lambda\bar T}/\pi$ allowance whenever the vehicle sits more than $Ae^{-\lambda\bar T}/(\pi k)=5.7$~cm from its estimated target ($\lambda=2$), while the steady orbit radius of \eqref{eq:controller} about $\hat p$ at full amplitude is $A/\sqrt{k^2+\omega^2}=19.5$~cm; the amplitude cancels from that comparison, so no gain retuning escapes it: an unprojected law satisfies the allowance in steady orbit only for $\pi k\le\sqrt{k^2+\omega^2}\,e^{-\lambda\bar T}$, i.e.\ $k\le0.127$~s$^{-1}$ at these $\omega$ and $\lambda$, and raising $\omega$ instead at $k=1.2$ would need $\omega\ge4.26$~rad/s, which violates (iii).
Projecting only the component that opposes $n(t)$ is therefore the minimal intervention that restores the guarantee at the implemented gain: it leaves lateral and aligned seeking untouched, and it engages exactly where the unprojected law would break the half-period displacement inequality.
Measured packet by packet in the no-transient supervised run, the projection clips $15$ of the $31$ underexcited packets, the measured $\|q-\hat p\|$ grows from zero to a $23.6$~cm plateau as the swirl orbit opens (wider than the unprojected $19.5$~cm steady orbit, because the clipped pull no longer drags the loop closed), and the certificate clears $\bar S=0.16$ at packet $32$ ($2.5$~s).
At $\lambda=2$ the bound evaluates to $\delta=0.237$ and $T^\star\le55.9$~s ($698$ packets), so acquisition completes a factor of $22$ inside it; the remaining conservatism is structural, since the proof credits only two stored packets per half-period, while the loop stores $\pi/(\omega\bar T)\approx87$ and clears the threshold before a single half-period completes.
The finitely-many-resets hypothesis that Proposition~\ref{prop:acquisition} feeds the recovery result is also verified directly in the logs: $31$ resets, then none, with $S_v$ monotone thereafter.
The archived code and data include the per-packet logs and a script that reproduces every number in this check.

\begin{proposition}[Closed-loop recovery]
\label{prop:closedloop}
Suppose Algorithm~\ref{alg:supervised} uses $\lambda>0$, performs finitely many resets, and the stored window eventually satisfies $S_v\ge\bar S>0$.
Then the rank condition of Fact~\ref{fact:constructive} holds and the constructive seed is well defined.
In the noiseless case, if the Gauss--Newton iterate enters the local basin around the true state, then $\hat p(t)\to p$ by local convergence of the damped Gauss--Newton iteration~\cite{nocedal2006numerical}; substituting this into the tracking-error dynamics derived below and invoking input-to-state stability then gives $q(t)\to p$~\cite{sontag1995characterizations,khalil2002nonlinear}.
If instead $\|\hat p-p\|\le\bar e$ after the final reset, then $q(t)$ converges to the ball $\|q-p\|\le\bar e$ for the unit-DC-gain first-order seeking loop.
\end{proposition}

\begin{proof}
Since $S_v>0$, Lemma~\ref{lem:spread} gives two distinct local vehicle vectors and a nonzero difference vector, so Fact~\ref{fact:constructive} applies and the estimator's rank condition holds.
Let $e=q-p$ after the final reset; the certificate has cleared, so the projection \eqref{eq:projection} is disengaged and the seeking law is exactly $-k(q-\hat p)$. Substituting $\hat p=p+(\hat p-p)$ into the control law gives the tracking-error dynamics
\begin{equation}
\dot e=-ke+k(\hat p-p)+u^{\rm exp}(t). \label{eq:trackingerror}
\end{equation}
The unforced system $\dot e=-ke$ is exponentially stable, so \eqref{eq:trackingerror} is ISS in the input $(\hat p-p,u^{\rm exp})$~\cite{sontag1995characterizations,khalil2002nonlinear}.
If $\hat p-p\to0$ and $u^{\rm exp}(t)\to0$ (the envelope decays exponentially after the final reset), the converging-input property of ISS gives $e\to0$, i.e., $q\to p$.
If instead $\|\hat p-p\|\le\bar e$ is only bounded, the same linear, unit-DC-gain system gives $\limsup_t\|e(t)\|\le\bar e$.
\end{proof}

To our knowledge, Algorithm~\ref{alg:supervised} is the first closed-loop hidden-target-seeking method for a single unknown-pose range-bearing relay in which the same trajectory-spread certificate quantifies calibration accuracy, triggers exploration online, and, under the stated sampling and cumulative-window assumptions, is guaranteed to cross any prescribed threshold in finite time under the implemented sample-and-hold controller. The control contribution is the certificate-supervised hybrid law, particularly the underexcited-phase projection \eqref{eq:projection} that makes the online trigger and finite-acquisition proof refer to the controller as implemented, without requiring online trajectory optimization or a matrix-valued observability objective.
The seeking law could equally be replaced by a receding-horizon planner that carries the estimate's remaining uncertainty into the motion constraints, as in chance-constrained motion planning for uncertain multi-agent systems~\cite{bagla2019receding}; the supervision layer is agnostic to that choice.

\section{Validation}
\label{sec:validation}

All simulations use the same C++ estimator core as the identifiability study of~\cite{bagla2027identifiability}, so the closed-loop results below are consistent with its open-loop Monte Carlo validation; Section~\ref{sec:gazebo} additionally runs the same supervised loop software-in-the-loop through a ROS~2/Gazebo stack.
Throughout this section the supervisor runs with $\bar S=0.16$, selected by \eqref{eq:nativethreshold} from the accuracy criterion these experiments declare as success.
The flagship trajectory has $r_{\max}=4.08$~m, hence $r_{\max}\sigma_\theta=0.0163$~m is below $\sigma_r=0.02$~m and the conservative first-order noise proxy is $\sigma_{\rm eff}=0.02$~m; therefore $\bar S=\sigma_{\rm eff}^2/\epsilon_\psi^2=0.16$ targets $\epsilon_\psi=0.05$~rad.
Because Corollary~\ref{cor:variance} is a variance statement, $\epsilon_\psi$ is a population target rather than a per-trial bound: aggregate results report the across-trial root-mean-square error $(\frac1M\sum_j\mathrm{wrap}(\hat\psi_j-\psi)^2)^{1/2}$ with percentile-bootstrap 95\% confidence intervals over $M=100$ paired trials (both policies consume the same noise realization per trial), plus the fraction of trials meeting the same fixed 0.05-rad level per trial; the design rule promises the former, not the latter.

\subsection{Closed-Loop Simulation}
\label{sec:flagship}

Algorithm~\ref{alg:supervised} was run in closed loop for 120 packets under the kinematic model of \eqref{eq:controller}.
At the first packet, the window has not yet accumulated spread, and the errors reflect that: target error $1.415$ m, vehicle-to-goal distance $5.372$ m, beacon-position error $4.427$ m, and beacon-yaw error $0.750$ rad.
Figure~\ref{fig:closedloop_traj} plots the resulting trajectory against the true target and the final estimate.

\begin{figure}[t]
\centering
\begin{tikzpicture}
\begin{axis}[
    width=\columnwidth,
    axis equal image,
    xlabel={$x$ (m)}, ylabel={$y$ (m)},
    grid=major, grid style={black!10},
    tick label style={font=\footnotesize},
    label style={font=\footnotesize},
    xmin=-3.7, xmax=2.7, ymin=-1.35, ymax=3.35,
]
\addplot[figblue, thick] table[x=x, y=y] {closed_loop_run.dat};
\addplot[only marks, mark=*, mark size=1.8pt, black] coordinates {(-3.0,2.6)};
\addplot[only marks, mark=*, mark size=2.2pt, figred] coordinates {(1.2,-0.75)};
\addplot[only marks, mark=x, mark size=3.2pt, ultra thick, figgreen] coordinates {(1.2036,-0.7485)};
\node[font=\footnotesize, figblue, anchor=south west] at (axis cs:-1.05,1.25) {$q_k$};
\node[font=\footnotesize, anchor=west] at (axis cs:-2.85,2.55) {$q_0$};
\node[font=\footnotesize, figred, anchor=south west] at (axis cs:1.32,-0.72) {$p$};
\node[font=\footnotesize, figgreen, anchor=north west] at (axis cs:1.32,-0.85) {$\hat p(60)$};
\end{axis}
\end{tikzpicture}
\caption{Closed-loop trajectory for the first 60 packets under Algorithm~\ref{alg:supervised}: the vehicle path $q_k$ from the initial pose $q_0$, the true target $p$, and the final estimate $\hat p(60)$.}
\label{fig:closedloop_traj}
\end{figure}
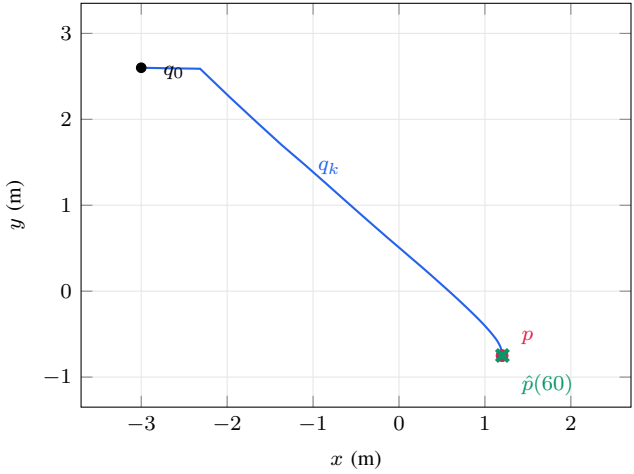

The estimator is deliberately inactive at packet one because the window has only one view; the prior target estimate is retained instead of feeding a gauge-dependent solution into the controller.
At packet two the constructive initializer becomes available, the target error falls from $1.415$~m to $0.249$~m, and Gauss--Newton refinement is enabled.
The supervisor retriggers at the first two packets, during which the projection \eqref{eq:projection} clips the seeking input: the stored-window spread is $0.0475$ after packet two, below $\bar S=0.16$, and reaches $0.236$ at packet three, after which the certificate stays cleared and the nominal law resumes.
The transit, not the excitation, supplies this spread: with the excitation term disabled entirely the same run still certifies at packet three ($S_v=0.37$ versus $0.24$) at the same final accuracy, hence the nearly straight path in Figure~\ref{fig:closedloop_traj}, whose excitation contributes only $0.53$~m of the $5.6$~m path. Section~\ref{sec:baseline} isolates the regime where the excitation term is load-bearing.
By packet 60, target error is $0.0039$~m, vehicle-to-goal distance $0.0417$~m, beacon-position error $0.0032$~m, and beacon-yaw error $0.00087$~rad; by packet 120 the corresponding errors are $0.0026$, $0.0023$, $0.0043$~m, and $0.00119$~rad.
Figure~\ref{fig:closedloop_err} plots these error traces alongside the accumulated spread $S_v$: the $S_v$ trace is monotone increasing while the excitation envelope is active and levels off once target seeking dominates, illustrating the supervision rule of Algorithm~\ref{alg:supervised} rather than a fixed exploration schedule.

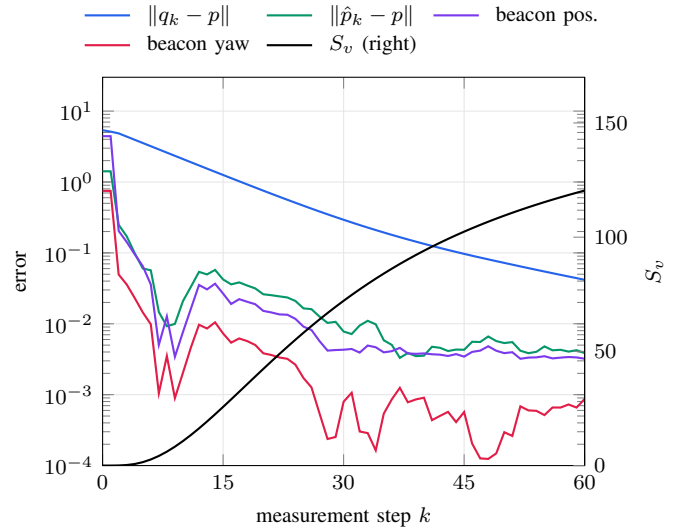
\begin{figure}[t]
\centering
\begin{tikzpicture}
\begin{semilogyaxis}[
    scale only axis,
    width=0.72\columnwidth, height=0.58\columnwidth,
    xlabel={measurement step $k$}, ylabel={error},
    xmin=0, xmax=60, ymin=1e-4, ymax=30,
    xtick={0,15,30,45,60},
    grid=major, grid style={black!10},
    tick label style={font=\footnotesize},
    label style={font=\footnotesize},
    legend columns=3,
    legend cell align=left,
    legend style={at={(0.5,1.03)}, anchor=south, draw=none, fill=none,
        font=\footnotesize, column sep=4pt},
]
\addplot[figblue, thick] table[x=step, y=goal] {closed_loop_run.dat};
\addlegendentry{$\|q_k-p\|$}
\addplot[figgreen, thick] table[x=step, y=target] {closed_loop_run.dat};
\addlegendentry{$\|\hat p_k-p\|$}
\addplot[figpurple, thick] table[x=step, y=bpos] {closed_loop_run.dat};
\addlegendentry{beacon pos.}
\addplot[figred, thick] table[x=step, y=byaw] {closed_loop_run.dat};
\addlegendentry{beacon yaw}
\addlegendimage{black, thick}
\addlegendentry{$S_v$ (right)}
\end{semilogyaxis}
\begin{axis}[
    scale only axis,
    width=0.72\columnwidth, height=0.58\columnwidth,
    axis y line*=right, axis x line=none,
    xmin=0, xmax=60, ymin=0, ymax=170,
    ylabel={$S_v$},
    ytick={0,50,100,150},
    tick label style={font=\footnotesize},
    label style={font=\footnotesize},
]
\addplot[black, thick] table[x=step, y=spread] {closed_loop_run.dat};
\end{axis}
\end{tikzpicture}
\caption{Target-seeking error $\|q_k-p\|$, target-estimate error $\|\hat p_k-p\|$, and beacon position (m) and yaw (rad) errors on the left log axis, with the accumulated spread $S_v$ on the right linear axis, over the flagship run's first 60 packets.}
\label{fig:closedloop_err}
\end{figure}

\subsection{ROS~2/Gazebo Software-in-the-Loop Experiment}
\label{sec:gazebo}

The same supervised loop also runs software-in-the-loop rather than inside the batch simulator.
Algorithm~\ref{alg:supervised} executes as a ROS~2 node that publishes planar velocity commands over \texttt{ros\_gz\_bridge} to a vehicle model whose motion is integrated by the Gazebo physics engine under zero-order-hold actuation (a velocity-control system plugin applies the held command each physics step; gravity and collisions are disabled on the vehicle link so the plant is a clean kinematic integrator, matching the paper's single-integrator vehicle model).
The estimator consumes the vehicle poses Gazebo publishes back as odometry, range-bearing packets emulated at those poses are delivered to the estimator one control cycle (80~ms) late as an explicit sensing delay, and the node paces its 80~ms packet cadence on the bridged simulation clock, so the experiment timeline is simulation time.
Estimation, the spread certificate $S_v$, and the supervision rule are byte-for-byte the same library code the batch simulator links; what changes is the plant: real message transport, actuation hold, and sensing latency replace the idealized explicit-Euler kinematics of \eqref{eq:controller} used for Figure~\ref{fig:closedloop_err}.
The node also moves translucent estimate markers in the Gazebo scene through the bridged \texttt{set\_pose} service each packet, so recordings of the run show the estimator converging live.

\begin{table}[t]
\centering
\caption{Gazebo software-in-the-loop repeatability.}
\label{tab:gazebo_batch}
\scriptsize
\setlength{\tabcolsep}{3pt}
\begin{tabular}{l c c c c c}
\hline
Run & Goal (m) & Target (m) & Bcn.\ pos.\ (m) & Bcn.\ yaw (rad) & Retrig.\\
\hline
seed 7 & 0.0026 & 0.0026 & 0.0043 & 0.0012 & 4\\
seed 8 & 0.0023 & 0.0032 & 0.0021 & 0.0001 & 4\\
seed 9 & 0.0030 & 0.0014 & 0.0030 & 0.0001 & 4\\
seed 10 & 0.0044 & 0.0015 & 0.0048 & 0.0010 & 4\\
seed 11 & 0.0080 & 0.0051 & 0.0038 & 0.0001 & 4\\
seed 12 & 0.0076 & 0.0046 & 0.0028 & 0.0008 & 4\\
seed 13 & 0.0016 & 0.0044 & 0.0042 & 0.0010 & 4\\
seed 14 & 0.0049 & 0.0018 & 0.0021 & 0.0009 & 4\\
seed 15 & 0.0034 & 0.0009 & 0.0018 & 0.0006 & 4\\
seed 16 & 0.0038 & 0.0024 & 0.0040 & 0.0004 & 4\\
\hline
across-seed RMSE & 0.0046 & 0.0031 & 0.0034 & 0.0007 & ---\\
\hline
delay 0 (seed 7) & 0.0023 & 0.0026 & 0.0041 & 0.0011 & 2\\
delay 2 (seed 7) & 0.0028 & 0.0025 & 0.0045 & 0.0013 & 5\\
\hline
\end{tabular}\\[3pt]
\parbox{\columnwidth}{\raggedright\scriptsize Note: final errors after 120 packets for ten seeds at one-packet delay, plus delay variants at seed 7. Delay-1 runs initialize at packet 3 and retrigger packets 1--4; delay 0 initializes at packet 2 and matches the simulator's packets-1--2 pattern.}
\end{table}

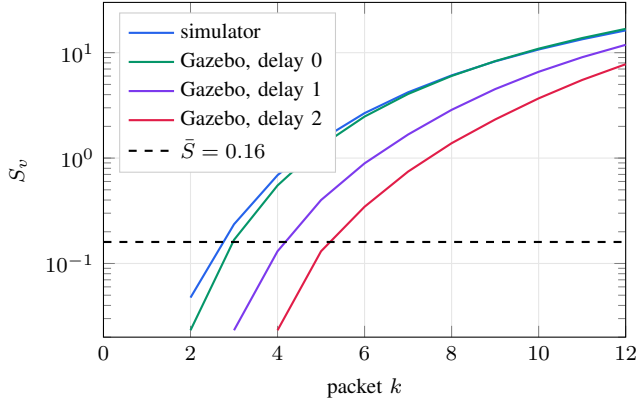
\begin{figure}[t]
\centering
\begin{tikzpicture}
\begin{semilogyaxis}[
    scale only axis,
    width=0.78\columnwidth, height=0.5\columnwidth,
    xlabel={packet $k$}, ylabel={$S_v$},
    xmin=0, xmax=12, ymin=0.02, ymax=30,
    xtick={0,2,4,6,8,10,12},
    grid=major, grid style={black!10},
    tick label style={font=\footnotesize},
    label style={font=\footnotesize},
    legend cell align=left,
    legend pos=north west,
    legend style={draw=black!20, fill=white, font=\footnotesize},
]
\addplot[figblue, thick] table[x=step, y=sim] {spread_delay_zoom.dat};
\addlegendentry{simulator}
\addplot[figgreen, thick] table[x=step, y=d0] {spread_delay_zoom.dat};
\addlegendentry{Gazebo, delay 0}
\addplot[figpurple, thick] table[x=step, y=d1] {spread_delay_zoom.dat};
\addlegendentry{Gazebo, delay 1}
\addplot[figred, thick] table[x=step, y=d2] {spread_delay_zoom.dat};
\addlegendentry{Gazebo, delay 2}
\addplot[black, dashed, thick] coordinates {(0,0.16) (12,0.16)};
\addlegendentry{$\bar S=0.16$}
\end{semilogyaxis}
\end{tikzpicture}
\caption{Sensing delay shifts certification: the simulator and delay-0 Gazebo run cross $\bar S=0.16$ at packet 3, while delays of one and two packets cross at packets 5 and 6. Delay postpones the two-view seed and spread certificate without changing the final state.}
\label{fig:spread_delay}
\end{figure}

Table~\ref{tab:gazebo_batch} and Figure~\ref{fig:spread_delay} quantify the run and its repeatability.
The delay shifts when the certificate clears but not where the loop ends: delay 0 reproduces the simulator's packet-3 crossing, while one- and two-packet delays move it to packets 5 and 6 (Figure~\ref{fig:spread_delay}), one packet later than without the projection because the delayed window certifies on clipped-transit poses.
The flagship Gazebo run ends at $2.6$~mm goal distance, $2.6$~mm target error, $4.3$~mm beacon-position error, and $1.2$~mrad yaw error; across ten seeds the corresponding RMSEs are $4.6$, $3.1$, $3.4$~mm, and $0.7$~mrad.
Before the two-view seed becomes available the node leaves $\hat p$ at its prior, so sensing delay creates a transparent initialization plateau rather than a gauge-dependent control transient.

\subsection{Baseline Comparison Against a Fixed Excitation Schedule}
\label{sec:baseline}

\begin{figure*}[t]
\centering
\begin{tikzpicture}
\begin{axis}[
    scale only axis,
    width=0.23\textwidth,
    axis equal image,
    xlabel={$x$ (m)}, ylabel={$y$ (m)},
    grid=major, grid style={black!10},
    tick label style={font=\footnotesize},
    label style={font=\footnotesize},
    xmin=1.12, xmax=1.44, ymin=-0.82, ymax=-0.54,
    legend cell align=left,
    legend columns=2,
    legend style={at={(0.5,1.03)}, anchor=south, draw=none, fill=none,
        font=\footnotesize, column sep=5pt},
]
\addplot[figpurple, thick] table[x=x, y=y] {notransient_fixed.dat};
\addlegendentry{fixed schedule}
\addplot[figblue, thick] table[x=x, y=y] {notransient_supervised.dat};
\addlegendentry{supervised}
\addplot[only marks, mark=*, mark size=2.2pt, figred] coordinates {(1.2,-0.75)};
\node[font=\footnotesize, figred, anchor=east] at (axis cs:1.19,-0.75) {$p$};
\end{axis}
\end{tikzpicture}\hfill
\begin{tikzpicture}
\begin{semilogyaxis}[
    scale only axis,
    width=0.44\textwidth, height=0.20\textwidth,
    xlabel={packet $k$}, ylabel={beacon-yaw error (rad)},
    xmin=0, xmax=120, ymin=5e-5, ymax=3,
    grid=major, grid style={black!10},
    tick label style={font=\footnotesize},
    label style={font=\footnotesize},
    legend cell align=left,
    legend columns=4,
    legend style={at={(0.5,1.03)}, anchor=south, draw=none, fill=none,
        font=\footnotesize, column sep=5pt},
]
\addplot[figpurple, thick] table[x=step, y=byaw] {notransient_fixed.dat};
\addlegendentry{fixed schedule}
\addplot[figgreen, thick] table[x=step, y=byaw] {notransient_information.dat};
\addlegendentry{information-gradient}
\addplot[figblue, thick] table[x=step, y=byaw] {notransient_supervised.dat};
\addlegendentry{supervised}
\addplot[figred, dashed, semithick] coordinates {(0,0.05) (120,0.05)};
\addlegendentry{$\epsilon_\psi=0.05$}
\addplot[black!50, dashed, semithick] coordinates {(32,5e-5) (32,3)};
\node[font=\footnotesize, black!60, anchor=west, rotate=90] at (axis cs:32,3e-3) {$S_v\ge\bar S$};
\end{semilogyaxis}
\end{tikzpicture}
\caption{The hardest tested regime (no transient, fastest decay $\lambda=2$~s$^{-1}$), one paired run: the vehicle starts at the true target with a perfect prior, so only excitation moves it. Left: the paths. The fixed schedule's excitation dies before certifying, while supervision holds its loop until $S_v\ge\bar S$ (the information-gradient path coincides with the fixed one). Right: the yaw error. The unsupervised schedules plateau at $0.093$ and $0.109$~rad, while supervision certifies at packet 32 and ends at $0.0070$~rad, a $13\times$ margin at identical task error. Population statistics: Table~\ref{tab:lambda_sweep}.}
\label{fig:showcase}
\end{figure*}
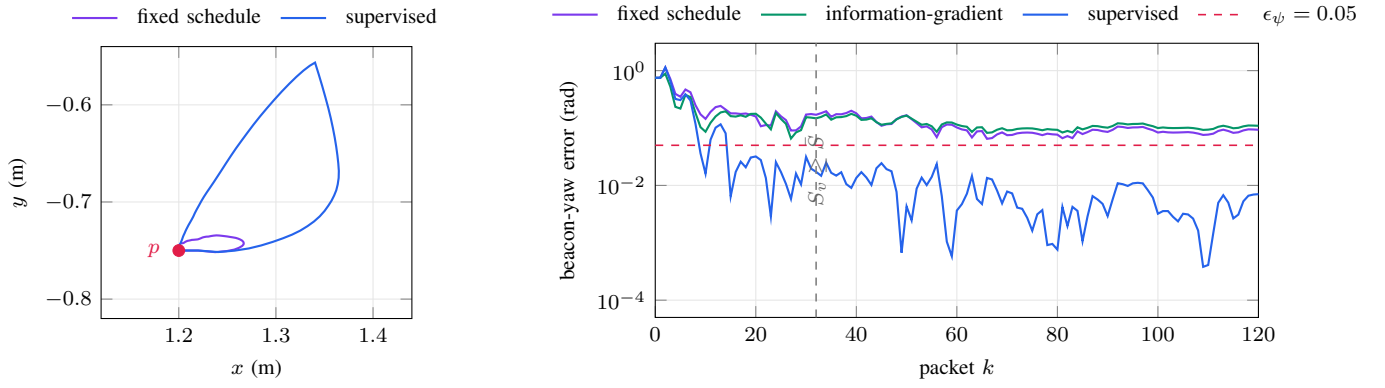

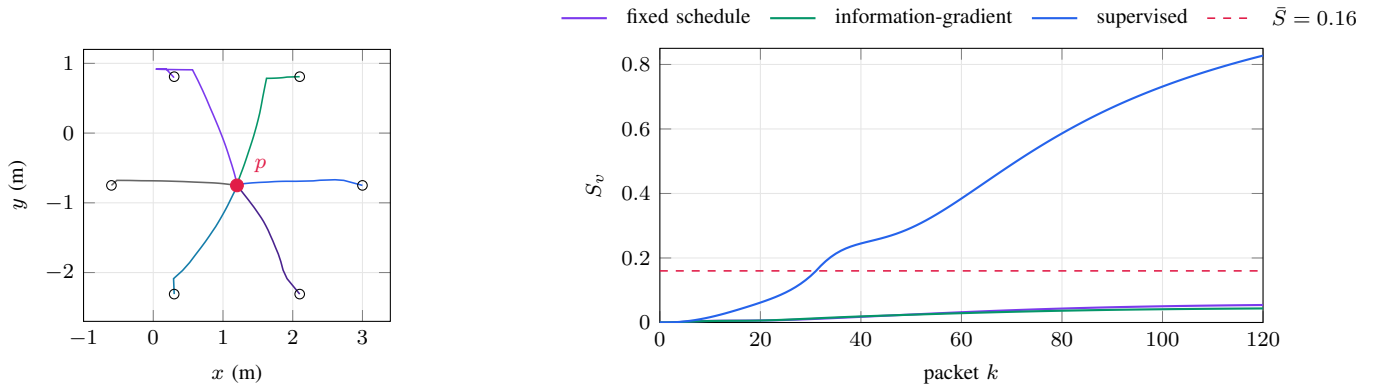
\begin{figure*}[t]
\centering
\begin{tikzpicture}
\begin{axis}[
    scale only axis,
    width=0.23\textwidth,
    axis equal image,
    xlabel={$x$ (m)}, ylabel={$y$ (m)},
    grid=major, grid style={black!10},
    tick label style={font=\footnotesize},
    label style={font=\footnotesize},
    xmin=-1.0, xmax=3.4, ymin=-2.7, ymax=1.2,
]
\addplot[figblue, semithick] table[x=x, y=y] {seek_ring_0.dat};
\addplot[figgreen, semithick] table[x=x, y=y] {seek_ring_1.dat};
\addplot[figpurple, semithick] table[x=x, y=y] {seek_ring_2.dat};
\addplot[black!60, semithick] table[x=x, y=y] {seek_ring_3.dat};
\addplot[figblue!50!figgreen, semithick] table[x=x, y=y] {seek_ring_4.dat};
\addplot[figpurple!60!black, semithick] table[x=x, y=y] {seek_ring_5.dat};
\addplot[only marks, mark=o, mark size=1.8pt, black] coordinates
    {(3.0,-0.75) (2.1,0.809) (0.3,0.809) (-0.6,-0.75) (0.3,-2.309) (2.1,-2.309)};
\addplot[only marks, mark=*, mark size=2.4pt, figred] coordinates {(1.2,-0.75)};
\node[font=\footnotesize, figred, anchor=south west] at (axis cs:1.32,-0.70) {$p$};
\end{axis}
\end{tikzpicture}\hfill
\begin{tikzpicture}
\begin{axis}[
    scale only axis,
    width=0.44\textwidth, height=0.20\textwidth,
    xlabel={packet $k$}, ylabel={$S_v$},
    xmin=0, xmax=120, ymin=0, ymax=0.85,
    grid=major, grid style={black!10},
    tick label style={font=\footnotesize},
    label style={font=\footnotesize},
    legend cell align=left,
    legend columns=4,
    legend style={at={(0.5,1.03)}, anchor=south, draw=none, fill=none,
        font=\footnotesize, column sep=5pt},
]
\addplot[figpurple, thick] table[x=step, y=spread] {notransient_fixed.dat};
\addlegendentry{fixed schedule}
\addplot[figgreen, thick] table[x=step, y=spread] {notransient_information.dat};
\addlegendentry{information-gradient}
\addplot[figblue, thick] table[x=step, y=spread] {notransient_supervised.dat};
\addlegendentry{supervised}
\addplot[figred, dashed, semithick] coordinates {(0,0.16) (120,0.16)};
\addlegendentry{$\bar S=0.16$}
\end{axis}
\end{tikzpicture}
\caption{Variety and mechanism behind Figure~\ref{fig:showcase}. Left: six supervised target-seeking runs from a ring of starts (radius $1.8$~m, distinct seeds, common wrong prior); every run certifies by packet ten and ends below $0.0045$~rad yaw error, so certification is not tied to a chosen start. Right: the spread certificate of the paired no-transient run; the fixed and information-gradient schedules stall below $\bar S=0.16$, while the supervised policy holds its excitation loop until the certificate clears at packet 32.}
\label{fig:variety}
\end{figure*}

\begin{table*}[t]
\centering
\caption{Decay-rate sweep.}
\label{tab:lambda_sweep}
\small
\resizebox{\textwidth}{!}{%
\begin{tabular}{c c c c c c c}
\hline
& \multicolumn{3}{c}{Fixed circular} & \multicolumn{3}{c}{Supervised (Alg.~\ref{alg:supervised})}\\
\cline{2-4}\cline{5-7}
Decay rate $\lambda$ & Bcn.\ pos.\ RMSE (m) & Bcn.\ yaw RMSE (rad) & Yaw succ. & Bcn.\ pos.\ RMSE (m) & Bcn.\ yaw RMSE (rad) & Yaw succ.\\
\hline
0.02 & 0.034 [0.030, 0.039] & 0.010 [0.009, 0.011] & 1.00 & 0.033 [0.028, 0.037] & 0.009 [0.008, 0.011] & 1.00\\
0.05 & 0.039 [0.034, 0.044] & 0.011 [0.010, 0.013] & 1.00 & 0.035 [0.031, 0.039] & 0.010 [0.009, 0.011] & 1.00\\
0.10 & 0.048 [0.042, 0.053] & 0.014 [0.012, 0.015] & 1.00 & 0.039 [0.034, 0.043] & 0.011 [0.010, 0.013] & 1.00\\
0.25 & 0.072 [0.063, 0.081] & 0.021 [0.018, 0.023] & 0.99 & 0.049 [0.043, 0.054] & 0.014 [0.012, 0.016] & 1.00\\
0.50 & 0.105 [0.091, 0.118] & 0.030 [0.026, 0.034] & 0.90 & 0.058 [0.051, 0.064] & 0.016 [0.014, 0.018] & 1.00\\
1.00 & 0.150 [0.130, 0.169] & 0.043 [0.037, 0.049] & 0.73 & 0.064 [0.056, 0.071] & 0.018 [0.016, 0.020] & 1.00\\
2.00 & 0.226 [0.195, 0.253] & 0.065 [0.056, 0.073] & 0.56 & 0.067 [0.058, 0.075] & 0.019 [0.017, 0.021] & 1.00\\
\hline
\end{tabular}}\\[3pt]
\parbox{\textwidth}{\raggedright\scriptsize Note: $\lambda$ is in s$^{-1}$; no-transient scenario, $M=100$ paired trials. Entries are final-error RMSEs after 120 packets with percentile-bootstrap 95\% CIs; yaw success uses the fixed 0.05-rad criterion. Target RMSE remains $3.4$--$3.6$~mm for both policies; supervision retriggers $30.8$--$31.7$ packets on average across the sweep.}
\end{table*}

Algorithm~\ref{alg:supervised} differs from the fixed decaying-swirl envelope used above in two ways: the fixed schedule sets $t_0=0$ once and never revisits it, while the supervised controller resets $t_0$ whenever $S_v(\mathcal{W})<\bar S$; and while the certificate remains below threshold, the supervised controller applies the projected seeking law \eqref{eq:projection}, so hypothesis (ii) of Proposition~\ref{prop:acquisition} holds by construction rather than by assumption.
The two are compared directly, holding the estimator, noise, initial condition, and integration step fixed.
In the nominal scenario of Figure~\ref{fig:closedloop_err}, where the vehicle starts $5.372$~m from the target and the transient supplies substantial spread, the two controllers remain close: final beacon-position and yaw errors are $0.0048$~m and $0.0014$~rad for the fixed schedule against $0.0054$~m and $0.0016$~rad for the supervised loop, whose two resets clip only the first two packets of the transit.
Supervision costs almost nothing when the fixed schedule happens to be adequate.

Figure~\ref{fig:showcase} makes the hardest regime the sweep contains concrete (no convergence transient, fastest tested decay) with three policies consuming identical packets: the fixed decaying swirl, an information-gradient schedule that ascends the log-determinant of the calibration information but shares the same decaying envelope, and the supervised policy.
Both unsupervised schedules let excitation die before the certificate clears and plateau at $0.093$ and $0.109$~rad yaw error ($0.32$ and $0.38$~m relay position); the supervised policy holds its excitation loop at full amplitude until $S_v\ge\bar S$ at packet 32 and ends at $0.0070$~rad and $0.0240$~m, at the same $2$--$3$~mm task error.
The failure is thus a property of unsupervised scheduling as such, not of the excitation shape: even an information-seeking law fails once its envelope, rather than a measured certificate, decides when exploration stops.
Figure~\ref{fig:variety} adds the certificate trace behind that run and the complementary variety result: six supervised seeking runs from a ring of start positions all certify en route within ten packets, so the near-straight transits of the nominal scenario are the incidental-excitation mechanism of Section~\ref{sec:seeking} at work rather than a curated geometry.

To test robustness to schedule mismatch rather than a single hand-selected failure case, the decay rate $\lambda$ is swept over $\{0.02,0.05,0.10,0.25,0.50,1.0,2.0\}$ in a scenario with no convergence transient: the vehicle is initialized at the true target, so only the excitation schedule itself can excite the geometry.
Table~\ref{tab:lambda_sweep} reports across-trial RMSE for both policies over $M=100$ paired trials of the same $120$-packet run.

The sweep confirms the expected asymmetry at the population level.
The fixed schedule degrades monotonically as its decay outruns the time-to-adequate-excitation: beacon-position RMSE grows from $0.034$~m at $\lambda=0.02$ to $0.226$~m at $\lambda=2$, yaw RMSE from $0.010$ to $0.065$~rad, and success falls from $100\%$ to $56\%$, while target RMSE stays near $3.5$~mm.
Calibration failure is silent in the tracking metric, exactly the split Corollary~\ref{cor:variance} predicts: the target rows carry no information on the calibration state $\zeta$ once $p$ absorbs their variation, so a poorly excited window still fits $p$ passably while leaving $x$ and $\psi$ poorly identified.
The supervised controller also degrades mildly with faster post-certificate decay, but remains inside the design target: beacon-position RMSE is $0.033$--$0.067$~m, yaw RMSE $0.0095$--$0.0191$~rad, and all trials satisfy the yaw criterion.

Two features of Table~\ref{tab:lambda_sweep} tie back to the theory.
First, the supervised controller's excitation acquisition is nearly independent of the mismatch parameter: it retriggers for $30.8$--$31.7$ packets on average across the sweep, since each reset holds the envelope at full amplitude while the window is underexcited; the residual $\lambda$-dependence enters only through the projection level $Ae^{-\lambda\bar T}/\pi$ in \eqref{eq:projection}, consistent with the acquisition mechanism of Proposition~\ref{prop:acquisition}.
Second, $\bar S=0.16$ targets a population yaw RMSE of $0.05$~rad and the supervised policy remains below it at every $\lambda$ (at most $0.0192$~rad); the fixed schedule exceeds the target only at the fastest decay, exposing a graded rather than manufactured failure boundary.
Algorithm~\ref{alg:supervised} does not need to know the time-to-adequate-excitation in advance; it measures $S_v$ directly and keeps resetting until the certificate clears, at the excitation cost quantified next.

\subsection{Spread-Threshold Ablation}
\label{sec:ablation}

\begin{table*}[t]
\centering
\caption{Spread-threshold ablation.}
\label{tab:ablation}
\small
\resizebox{\textwidth}{!}{%
\begin{tabular}{c c c c c c c c c}
\hline
$\bar S$ & Pred.\ yaw (rad) & Yaw RMSE (rad) & Yaw succ. & Pkts.\ to $\bar S$ & Path (m) & Effort & Target RMSE (m) & Bcn.\ pos.\ RMSE (m)\\
\hline
0.05 & 0.089 & 0.0239 [0.0204, 0.0272] & 0.96 & 18.3 & 0.373 & 0.100 & 0.00309 [0.00278, 0.00337] & 0.0835 [0.0713, 0.0950]\\
0.16 & 0.050 & 0.0187 [0.0157, 0.0217] & 0.99 & 31.8 & 0.557 & 0.167 & 0.00318 [0.00285, 0.00349] & 0.0659 [0.0551, 0.0761]\\
1.0 & 0.020 & 0.0125 [0.0107, 0.0144] & 1.00 & 49.9 & 1.039 & 0.258 & 0.00332 [0.00296, 0.00368] & 0.0443 [0.0377, 0.0512]\\
9.04 & 0.0067 & 0.0053 [0.0046, 0.0059] & 1.00 & 109.2 & 1.768 & 0.554 & 0.00327 [0.00292, 0.00361] & 0.0186 [0.0162, 0.0209]\\
\hline
\end{tabular}}\\[3pt]
\parbox{\textwidth}{\raggedright\scriptsize Note: $M=100$ supervised trials, no-transient scenario, $\lambda=2$~s$^{-1}$; success is fixed at 0.05~rad. Pred.\ yaw is $\sigma_{\rm eff}/\sqrt{\bar S}$; brackets are percentile-bootstrap 95\% CIs; effort is $\sum_k\|u_k^{\rm exp}\|^2\Delta t$. Every threshold is reached within the 120-packet horizon.}
\end{table*}

Table~\ref{tab:ablation} sweeps the supervisor's only tuning constant over $\bar S\in\{0.05,0.16,1.0,9.04\}$ in the same no-transient scenario (fast decay $\lambda=2$, supervised policy only, $M=100$ trials), holding the trial success criterion fixed at 0.05~rad so that success means the same accuracy at every threshold.
At every threshold the measured yaw RMSE is below the native design value: $0.0239<0.089$, $0.0187<0.050$, $0.0125<0.020$, and $0.0053<0.0067$~rad for $\bar S=0.05,0.16,1.0,9.04$.
Accuracy is bought with excitation: increasing $\bar S$ from $0.05$ to $0.16$ raises mean certification time from $18.3$ to $31.8$ packets and effort from $0.100$ to $0.167$, while beacon-position RMSE improves from $0.0835$ to $0.0659$~m; target RMSE remains near $3$~mm.
The $\bar S=9.04$ row records a side effect of the projection in \eqref{eq:projection}: with the seeking pull clipped, the underexcited orbit opens rather than being dragged closed, and a threshold that was unreachable in 120 packets without the projection now certifies at packet 109 and delivers its predicted level.
A threshold remains meaningful only up to the spread reachable at the loop's excitation radius and packet budget, which is exactly what Corollary~\ref{cor:budget} prices: $\bar S=9.04$ is the high-precision operating point of Table~\ref{tab:budget}, here reached only at the edge of the horizon.
The success column also illustrates the population-versus-trial distinction: at $\bar S=0.16$ the population RMSE is $0.0187$~rad and $99\%$ of individual trials meet the separate 0.05-rad criterion.

\subsection{Target Seeking with Incidental Excitation}
\label{sec:seeking}

The sweep above isolates calibration by starting the vehicle at the true target, so that only excitation moves it.
The complementary scenario starts the vehicle at its wrong initial target estimate, $1.4$~m from the true target, with the fast decay $\lambda=2$: the seeking term $-k(q-\hat p)$ is initially quiescent (the vehicle sits at $\hat p$), and target-seeking success depends on the loop generating excitation and calibrating en route.

\begin{table}[t]
\centering
\caption{Target-seeking comparison.}
\label{tab:seeking}
\scriptsize
\setlength{\tabcolsep}{3pt}
\resizebox{\columnwidth}{!}{%
\begin{tabular}{l c c c c c}
\hline
Policy & Pkts.\ to goal & Goal RMSE (m) & Target RMSE (m) & Bcn.\ pos.\ RMSE (m) & Retrig.\\
\hline
Fixed & $34.3\pm0.3$ & 0.00312 [0.00285, 0.00338] & 0.00309 [0.00281, 0.00334] & 0.0159 [0.0134, 0.0183] & 0\\
Supervised & $36.1\pm0.3$ & 0.00312 [0.00285, 0.00339] & 0.00308 [0.00282, 0.00335] & 0.0159 [0.0136, 0.0181] & 5.5\\
\hline
\end{tabular}}\\[3pt]
\parbox{\columnwidth}{\raggedright\scriptsize Note: $M=100$ paired trials, $\lambda=2$~s$^{-1}$; the vehicle starts at its wrong initial target estimate, $1.4$~m from the true target. Pkts.\ to goal is mean $\pm$ 95\% CI over the trials that crossed the goal threshold (all did, under both policies); brackets are bootstrap 95\% CIs.}
\end{table}

Over $M=100$ paired trials both policies succeed in every trial with statistically indistinguishable final accuracy; yaw RMSE is $0.00472$ [$0.00398$, $0.00543$]~rad (fixed) and $0.00467$ [$0.00398$, $0.00531$]~rad (supervised).
This is the same self-excitation mechanism that makes the nominal scenario of Section~\ref{sec:flagship} easy: any trajectory with a substantial transit accumulates $S_v$ incidentally, and the certificate clears without dedicated exploration.
Ordinary closed-loop motion can deliver excitation naturally; supervision matters most when that incidental excitation is absent or insufficient, as in the no-transient sweep.
Here the measured price of supervision is arrival time, not accuracy: the fixed schedule reaches the goal in $34.3\pm0.3$ packets against $36.1\pm0.3$ supervised, a $5\%$ ($0.14$~s) cost from projecting the seeking input over its $5.5$ retriggered packets while the certificate is below threshold.

\subsection{Cross-Study Interpretation}

The three closed-loop studies isolate different causal roles of motion.
The no-transient decay sweep removes incidental seeking motion and therefore measures whether the excitation schedule alone can calibrate the relay.
The seeking study restores a large state transient and shows that ordinary regulation can supply the missing spread without a performance difference between policies.
The Gazebo delay study changes message timing while preserving the commanded geometry: matched seed-7 runs isolate the certification shift across zero-, one-, and two-packet delays, while a separate ten-seed experiment tests repeatability at a fixed one-packet delay.
Together these results delimit the claim: supervision is an online adequacy mechanism, not a universal source of lower tracking error when the nominal trajectory is already informative.

The experiments also expose why calibration must be evaluated separately from task completion.
In the no-transient sweep, target RMSE remains near $3.5$~mm while relay-position and yaw errors deteriorate sharply under the fixed schedule; a controller judged only by target or goal error would miss this loss of frame calibration.
Conversely, the supervised policy keeps the observable certificate above its design level without inspecting ground-truth estimation error.
This distinction is operationally important when the calibrated relay frame will later be reused for another target, shared with another vehicle, or carried into a downstream planner: task success in one episode does not certify that the recovered frame is trustworthy.

\section{Finite-Memory Deployment}

Algorithm~\ref{alg:supervised} uses a growing packet window, which makes $S_v$ nondecreasing and is the reason Proposition~\ref{prop:acquisition} can conclude that a cleared threshold remains cleared.
A bounded-memory implementation instead uses the most recent $L$ packets,
\[
S_v^{(L)}(k)=\sum_{j=k-L+1}^{k}\|q_j-\bar q_k^{(L)}\|^2,
\]
and must continue monitoring the certificate because informative poses can leave the window.

\begin{corollary}[Rolling-window calibration bound]
\label{cor:rolling}
Suppose every active $L$-packet window has nonzero relay ranges bounded by $r_{\max}$ and satisfies $S_v^{(L)}(k)\ge\bar S>0$.
Under the local native range-bearing model, the yaw information after eliminating relay translation is at least $\bar S/\sigma_{\rm eff}^2$, and therefore
\begin{equation}
\operatorname{var}(\hat\psi_k)\le\frac{\sigma_{\rm eff}^2}{\bar S},\qquad
\sigma_{\rm eff}=\max\{\sigma_r,r_{\max}\sigma_\theta\}. \label{eq:rollingbound}
\end{equation}
Thus the same accuracy-driven threshold used by the growing-window supervisor remains valid for every bounded-memory window that maintains the certificate.
\end{corollary}

\begin{proof}
The radial and tangential local-vector information eigenvalues are $\sigma_r^{-2}$ and $(r_k^v\sigma_\theta)^{-2}$.
The range bound makes each eigenvalue at least $\sigma_{\rm eff}^{-2}$.
Eliminating translation leaves the centered yaw energy of Lemma~\ref{lem:spread}, so the yaw Schur complement is at least $S_v^{(L)}(k)/\sigma_{\rm eff}^2\ge\bar S/\sigma_{\rm eff}^2$; inversion gives \eqref{eq:rollingbound}.
\end{proof}

A rolling-window supervisor can therefore preserve the same accuracy design rule by retriggering whenever the lower threshold is violated.
Using release and retrigger levels $\bar S_{\rm hi}>\bar S_{\rm lo}>0$ supplies hysteresis: exploration is released only above $\bar S_{\rm hi}$ and is not re-entered until the rolling certificate falls below $\bar S_{\rm lo}$.
The finite-acquisition proof applies to each re-entry episode while its assumptions hold, since each re-entry re-engages the projected seeking law \eqref{eq:projection} and hypothesis (ii) again holds by construction, but the growing-window conclusion of finitely many total resets no longer follows automatically; proving dwell time or eventual cessation for a sliding window requires an additional condition on the seeking trajectory.

\section{Statistical Validation Protocol}

The batch simulator and ROS~2 node link the same C++ estimator library, use native range-bearing residuals and analytic Jacobians, retain the prior before two-view initialization, and evaluate $t_k=(k-1)\Delta t$ with $\Delta t=0.08$~s.
The decay and seeking comparisons use $M=100$ paired trials: each fixed/supervised pair shares the world, initial condition, packet-noise realization, and integration step, so policy differences are not confounded by different random samples.
Reported intervals are percentile-bootstrap 95\% confidence intervals over trial-level terminal errors; success uses a fixed $0.05$-rad yaw criterion independent of solver termination, and threshold-ablation rows retain that criterion while varying only $\bar S$.
The Gazebo study separates transport effects from estimator changes: ten seeds use one-packet sensing delay, and matched seed-7 runs use delays of zero and two packets, while the estimator, controller gains, packet cadence, and noise model remain fixed.
The batch study therefore tests the statistical effect of supervision and threshold choice, whereas the software-in-the-loop study tests whether the same logic survives physics integration, zero-order-hold actuation, and delayed message delivery; neither is presented as physical-robot evidence.
The accompanying artifact supplies the deterministic experiment definitions, shared estimator source, trial-level logs, regression tests, and figure-generation workflow.
The identifiability analysis of~\cite{bagla2027identifiability} contains the gauge, constructive recovery, conditioning, and estimator theory; the present paper begins from that result and studies how a controller acquires and maintains the required finite-window excitation.

\section{Conclusion}

Building on the identifiability result of~\cite{bagla2027identifiability}, this paper showed that the same trajectory-spread margin that removes the self-calibration gauge also bounds constructive-seed error, decomposes target-estimate variance, and gives an explicit excitation budget for circular motion, with the spread threshold selected from a desired calibration-accuracy level rather than chosen heuristically.
An excitation-supervised algorithm uses these facts to retrigger exploration only when the stored window's spread certificate is below threshold, projecting the seeking input away from the excitation's push direction while it does so; the supervision rule provably acquires any required excitation within an explicit finite time, with every hypothesis of that guarantee satisfied by the controller as implemented, while, in the noiseless model with $\lambda>0$, local estimator convergence after certification yields target-seeking convergence.
A 100-trial paired decay-rate sweep showed the fixed schedule's calibration degrading monotonically, and silently in the target metric, as its decay outruns the unknown time-to-adequate-excitation, while the supervised controller remained inside its yaw-accuracy target at every decay rate; a threshold ablation validated the native range-bearing design rule at every tested threshold and priced tighter thresholds in excitation effort; and a target-seeking scenario showed when ordinary motion supplies sufficient incidental excitation, at a measured $5\%$ arrival cost.
A ROS~2/Gazebo software-in-the-loop experiment ran the supervised loop against a physics-integrated vehicle with real message transport and sensing delay: matched seed-7 runs shifted certification from packet 3 without delay to packet 6 with a two-packet delay while retaining similar final errors, and ten one-packet-delay seeds retained millimeter-scale position accuracy and sub-milliradian yaw RMSE.
The anchor-free extension in~\cite{bagla2026anchorfree} demonstrates that target seeking can remain well posed even when the globally referenced vehicle trajectory assumed here is unavailable, provided control and certification are carried out on the observable task quotient.
Remaining directions include a three-dimensional relay geometry, a greedy $S_v$-maximizing planner in place of the fixed circular excitation envelope, and physical robot validation of the supervision rule.

\bibliographystyle{IEEEtran}
\bibliography{references}

\begin{thebibliography}{10}
\providecommand{\url}[1]{#1}
\csname url@samestyle\endcsname
\providecommand{\newblock}{\relax}
\providecommand{\bibinfo}[2]{#2}
\providecommand{\BIBentrySTDinterwordspacing}{\spaceskip=0pt\relax}
\providecommand{\BIBentryALTinterwordstretchfactor}{4}
\providecommand{\BIBentryALTinterwordspacing}{\spaceskip=\fontdimen2\font plus
\BIBentryALTinterwordstretchfactor\fontdimen3\font minus
  \fontdimen4\font\relax}
\providecommand{\BIBforeignlanguage}[2]{{%
\expandafter\ifx\csname l@#1\endcsname\relax
\typeout{** WARNING: IEEEtran.bst: No hyphenation pattern has been}%
\typeout{** loaded for the language `#1'. Using the pattern for}%
\typeout{** the default language instead.}%
\else
\language=\csname l@#1\endcsname
\fi
#2}}
\providecommand{\BIBdecl}{\relax}
\BIBdecl

\bibitem{bagla2027identifiability}
Y.~Bagla, ``Trajectory-induced self-calibration for hidden-target localization
  through an unknown-pose range-bearing relay,'' 2026, arXiv:2608.09464.

\bibitem{bagla2026anchorfree}
------, ``Anchor-free hidden-target seeking via certified self-calibration
  under correlated odometry,'' 2026, arXiv:2609.27905.

\bibitem{boyd1986necessary}
S.~Boyd and S.~Sastry, ``Necessary and sufficient conditions for parameter
  convergence in adaptive control,'' \emph{Automatica}, vol.~22, no.~6, pp.
  629--639, 1986.

\bibitem{narendra1987persistent}
K.~S. Narendra and A.~M. Annaswamy, ``Persistent excitation in adaptive
  systems,'' \emph{Int. J. Control}, vol.~45, no.~1, pp. 127--160, 1987.

\bibitem{aranovskiy2017drem}
S.~Aranovskiy, A.~Bobtsov, R.~Ortega, and A.~Pyrkin, ``Performance enhancement
  of parameter estimators via dynamic regressor extension and mixing,''
  \emph{IEEE Trans. Autom. Control}, vol.~62, no.~7, pp. 3546--3550, 2017.

\bibitem{chowdhary2013concurrent}
G.~Chowdhary, T.~Yucelen, M.~M{\"u}hlegg, and E.~N. Johnson, ``Concurrent
  learning adaptive control of linear systems with exponentially convergent
  bounds,'' \emph{Int. J. Adapt. Control Signal Process.}, vol.~27, no.~4, pp.
  280--301, 2013.

\bibitem{krstic2000stability}
M.~Krsti{\'c} and H.-H. Wang, ``Stability of extremum seeking feedback for
  general nonlinear dynamic systems,'' \emph{Automatica}, vol.~36, no.~4, pp.
  595--601, 2000.

\bibitem{guler2017adaptive}
S.~G{\"u}ler, B.~Fidan, S.~Dasgupta, B.~D.~O. Anderson, and I.~Shames,
  ``Adaptive source localization based station keeping of autonomous
  vehicles,'' \emph{IEEE Trans. Autom. Control}, vol.~62, no.~7, pp.
  3122--3135, 2017.

\bibitem{hausman2017observability}
K.~Hausman, J.~A. Preiss, G.~S. Sukhatme, and S.~Weiss, ``Observability-aware
  trajectory optimization for self-calibration with application to {UAVs},''
  \emph{IEEE Robot. Autom. Lett.}, vol.~2, no.~3, pp. 1770--1777, 2017.

\bibitem{preiss2018simultaneous}
J.~A. Preiss, K.~Hausman, G.~S. Sukhatme, and S.~Weiss, ``Simultaneous
  self-calibration and navigation using trajectory optimization,'' \emph{Int.
  J. Robot. Res.}, vol.~37, no. 13--14, pp. 1573--1594, 2018.

\bibitem{peng2024trajectory}
J.~Peng, Q.~Wang, B.~Jin, Y.~Zhang, and K.~Lu, ``Trajectory optimization to
  enhance observability for bearing-only target localization and sensor bias
  calibration,'' \emph{Biomimetics}, vol.~9, no.~9, p. 510, 2024.

\bibitem{bishop2010optimality}
A.~N. Bishop, B.~Fidan, B.~D.~O. Anderson, K.~Do{\u{g}}an{\c{c}}ay, and P.~N.
  Pathirana, ``Optimality analysis of sensor-target localization geometries,''
  \emph{Automatica}, vol.~46, no.~3, pp. 479--492, 2010.

\bibitem{martinez2006optimal}
S.~Mart{\'\i}nez and F.~Bullo, ``Optimal sensor placement and motion
  coordination for target tracking,'' \emph{Automatica}, vol.~42, no.~4, pp.
  661--668, 2006.

\bibitem{zhou2011multirobot}
K.~Zhou and S.~I. Roumeliotis, ``Multirobot active target tracking with
  combinations of relative observations,'' \emph{IEEE Trans. Robot.}, vol.~27,
  no.~4, pp. 678--695, 2011.

\bibitem{bourgault2002information}
F.~Bourgault, A.~A. Makarenko, S.~B. Williams, B.~Grocholsky, and H.~F.
  Durrant-Whyte, ``Information based adaptive robotic exploration,'' in
  \emph{Proc. IEEE/RSJ Int. Conf. Intell. Robots Syst.}, 2002, pp. 540--545.

\bibitem{placed2023active}
J.~A. Placed, J.~Strader, H.~Carrillo, N.~Atanasov, V.~Indelman, L.~Carlone,
  and J.~A. Castellanos, ``A survey on active simultaneous localization and
  mapping: State of the art and new frontiers,'' \emph{IEEE Trans. Robot.},
  vol.~39, no.~3, pp. 1686--1705, 2023.

\bibitem{mesbah2018stochastic}
A.~Mesbah, ``Stochastic model predictive control with active uncertainty
  learning: {A} survey on dual control,'' \emph{Annu. Rev. Control}, vol.~45,
  pp. 107--117, 2018.

\bibitem{barshalom2001estimation}
Y.~Bar-Shalom, X.~R. Li, and T.~Kirubarajan, \emph{Estimation with Applications
  to Tracking and Navigation}.\hskip 1em plus 0.5em minus 0.4em\relax Wiley,
  2001.

\bibitem{nocedal2006numerical}
J.~Nocedal and S.~J. Wright, \emph{Numerical Optimization}, 2nd~ed.\hskip 1em
  plus 0.5em minus 0.4em\relax Springer, 2006.

\bibitem{sontag1995characterizations}
E.~D. Sontag and Y.~Wang, ``On characterizations of the input-to-state
  stability property,'' \emph{Syst. Control Lett.}, vol.~24, no.~5, pp.
  351--359, 1995.

\bibitem{khalil2002nonlinear}
H.~K. Khalil, \emph{Nonlinear Systems}, 3rd~ed.\hskip 1em plus 0.5em minus
  0.4em\relax Prentice Hall, 2002.

\bibitem{bagla2019receding}
Y.~Bagla and V.~Srivastava, ``On receding horizon chance constraint motion
  planning for uncertain multi-agent systems,'' in \emph{Proc. ASME Dyn. Syst.
  Control Conf.}, vol. 59162, 2019, p. V003T19A012.

\end{thebibliography}

\end{document}